\documentclass[12pt,draftcls,onecolumn]{IEEEtran}
\usepackage{cite}
\usepackage{amsmath,amssymb,amsfonts}
\usepackage{algorithmic}
\usepackage{graphicx}
\usepackage{textcomp}
\usepackage{stmaryrd}
\usepackage{amsfonts}
\usepackage{graphicx,times,amsmath,mathrsfs}
\usepackage{graphics,color}
\usepackage{graphicx}
\usepackage{latexsym}
\usepackage{amsmath}
\usepackage{amsfonts}
\usepackage{amssymb}
\usepackage{url}
\usepackage{subfigure}
\usepackage{times}
\usepackage{color}
\usepackage{float}
\usepackage{lettrine}
\newtheorem{theorem}{Theorem}

\newtheorem{definition}{Definition}
\newtheorem{lemma}{Lemma}
\newtheorem{remark}{Remark}

\newtheorem{assumption}{Assumption}

\allowdisplaybreaks
\def\BibTeX{{\rm B\kern-.05em{\sc i\kern-.025em b}\kern-.08em
    T\kern-.1667em\lower.7ex\hbox{E}\kern-.125emX}}
\begin{document}
\title{Prescribed-Time Synchronization of Multiweighted and Directed Complex Networks}
\author{Linlong Xu and Xiwei Liu,~\IEEEmembership{Senior Member, IEEE}

\thanks{This work was supported by the National Natural Science Foundation of China under Grant No. 62073243, 61673298; the Shanghai Municipal Science and Technology Major Project (2021SHZDZX0100), and the Fundamental Research Funds for the Central Universities.}
\thanks{The authors are with Department of Computer Science and Technology, Tongji University, and with the Key Laboratory of Embedded System and Service Computing, Ministry of Education, Shanghai 201804, China. Corresponding author: Xiwei Liu. (e-mail: xwliu@tongji.edu.cn).}
\thanks{}}

\maketitle

\begin{abstract}
In this note, we study the prescribed-time (PT) synchronization of multiweighted and directed complex networks (MWDCNs) via pinning control. Unlike finite-time and fixed-time synchronization, the time for synchronization can be preset as needed, which is independent of initial values and parameters like coupling strength. First and foremost, we reveal the essence of PT stability by improper integral, L'Hospital rule and Taylor expansion theory. Many controllers established previously for PT stability can be included in our new model. Then, we apply this new result on MWDCNs as an application. The synchronization error at the prescribed time is discussed carefully, so, PT synchronization can be reached. The network topology can be directed and disconnected, which means that the outer coupling matrices (OCMs) can be asymmetric and not connected. The relationships between nodes are allowed to be cooperative or competitive, so elements in OCMs and inner coupling matrices (ICMs) can be positive or negative. We use the rearranging variables' order technique to combine ICMs and OCMs together to get the sum matrices, which can make a bridge between multiweighted and single-weighted networks. Finally, simulations are presented to illustrate the effectiveness of our theory.
\end{abstract}


\section{Introduction}
\label{sec:introduction}
A complex dynamical network can be regarded as a graph with sundry nodes, complex structure, and divers connections. Over the past years, many scholars have conducted research on its dynamics such as consensus \cite{RR}, synchronization \cite{WT1} and control \cite{CTP1}, which means that every node in the network reaches the same state after a duration of time.

The settling time of synchronization is a key index to evaluate the network efficiency. This time means that synchronization should be achieved within it, which helps estimate the completion time of tasks. By now, this performance index has become a hot topic \cite{SPB}-\cite{BDN}. One type is the finite-time synchronization, which is an effective way to accelerate speed \cite{SPB}. It has advantages of good disturbance rejection and robustness against uncertainties, but the settling time heavily depends on initial values of the network. The other type is the fixed-time synchronization firstly proposed in \cite{PY1}, where the settling time can be independent on initial values. The essence of finite-time and fixed-time stability (or synchronization) was investigated in \cite{LLC2016,lxw2} by using the inverse function method.

Although the settling time of fixed-time stability is independent of initial states, it actually depends on the network parameters, such as coupling strengths, control gains, etc. Thus, the settling time for finite-time and fixed-time stability cannot be prescribed arbitrarily. The key difficulty lies that how to design a suitable controller to obtain prescribed-time (PT) stability. A new concept called ``prescribed convergence time'' was firstly proposed in \cite{LF1}, but it is often larger than the actual convergence time. To overcome this, Sanchez-Tones {\it{et al.}} \cite{JD1} put forward ``prescribed-time stability''. From then, some leading works have been presented, such as PT consensus \cite{WYJ1}-\cite{RYH1}, PT synchronization \cite{LD1}-\cite{lxy2} and control \cite{LWQ1}-\cite{DT2}. For consensus problem, Wang {\it{et al.}} \cite{WYJ1} presented a novel distributed protocol upon a new scaling function for multiagent systems. Based on a similar scaling function, Guo and Liang \cite{LJL1} designed the PT bipartite consensus protocol by using the time scale theory. Second-order multiagent systems were also investigated in \cite{RYH1}. For PT synchronization, an event-triggered control with a time-varying control gain was developed in \cite{LD1} and a smooth controller was designed in \cite{lxy1} for PT cluster synchronization. Shao {\it{et al.}} designed smooth controllers to achieve synchronization on Lur'e networks \cite{lxy2}. For stochastic nonlinear strict-feedback systems, Li and Krstic considered the PT mean-square stabilization in \cite{LWQ2} and PT output-feedback control in \cite{LWQ1}. Instead of fractional-power state feedback, \cite{SYD1}-\cite{SYD5} built a time-varying PT controller based on regular state feedback. Wang {\it{et al.}} \cite{WZW1} proposed the concept of ``practical prescribed-time stability'' and developed an adaptive fault-tolerant controller. Shakouri and Assadian \cite{SA1} defined ``triangular stability'' and proved that a PT controller can achieve it. Similar to \cite{KP1}, Tran and Yucelen \cite{DT1} established generalized time transformation functions and executed a control algorithm for stability analysis on perturbed systems, whereas time transformation was also applied on multiagent systems \cite{DT2}. Linear state and observer-based output feedbacks were designed for PT stabilization in \cite{ZB1}.

From the aforementioned papers, we can see that although PT stability or synchronization has been studied from lots of aspects, the essence of PT stability is still not revealed, and most of them only discuss networks with a single weight. Nodes may have many kinds of connections in practice, for instance, one can travel from one point to another point by railway, highway, ship, airplane, etc. Yao {\it{et al.}} \cite{YXQ1} considered the synchronization of fractional-order multiweighted complex networks, and Wang {\it{et al.}} \cite{WJL2} investigated the ${{H}}_{\infty}$ synchronization. For multiweighted and directed complex networks (MWDCNs), the design of Lyapunov function is a difficulty, and Liu \cite{lxw1} proposed two useful techniques to deal with this problem, whose routes were also adopted in \cite{wjq} and \cite{LSR2}.

The main contributions are listed as:

1. A new and general control scheme is designed to guarantee PT stability. The essence of PT stability is revealed through improper integral, L'Hospital rule and Taylor expansion theory. Time-varying functions are more general than many previous papers \cite{lxy1}-\cite{lxy2}, \cite{SYD1}-\cite{SYD5}, so that what they have established are included in our results. The control input is also proved to be bounded and be zero within the prescribed-time $T$.

2. We apply the obtain results on the PT synchronization problem for MWDCNs. In contrast to \cite{YXQ1} and \cite{WJL2}, where outer coupling matrices (OCMs) are required to be symmetric, we relax this requirement to be asymmetric and these OCMs are not necessarily strongly connected or even not connected. In addition to considering the cooperative relationship between nodes, we also consider the competitive relationship. It is better and has more application scenarios than \cite{lxw1,wjq,LSR2}.

3. Compared with many existing studies on finite-time and fixed-time synchronization, our work on PT synchronization is more challenging. The settling time is fully independent of initial states and network parameters. Rearranging variables' order technique is applied to obtain the sum (union) matrices by combining inner coupling matrices (ICMs) and OCMs.

In Section \ref{pre}, we give some definitions and lemmas. PT stability is carefully investigated, and its essence is revealed by the improper integral. In Section \ref{MWD}, PT synchronization for MWDCNs is investigated. Numerical simulations are given in Section \ref{simulation}. At last, we summarize this note in Section \ref{conclusion}.

\section{Preliminaries}\label{pre}
Let $\mathscr{G}=(\mathscr{X}, \mathscr{E}, M)$ be a directed graph, node set $\mathscr{X}= \{x_1, x_2, \cdots, x_N\}$, and denote $\Omega = \{1, 2, \cdots, N\}$. If an arc $(x_i, x_j)$ is in the set $\mathscr{E}$, then $M_{ij}\neq 0$ in the weighted adjacency matrix. For an undirected graph, $M_{ij}= M_{ji}$ always holds. For a directed graph (digraph), conversely, $M_{ij}$ and $M_{ji}$ may not be equal. In this note, we only consider digraphs, since undirected graph is just a special case of digraphs.

\begin{definition}(\cite{WT1})
An irreducible matrix $M = (M_{ij})\in R^{N \times N}$ is said to be a member
of \textbf{$A_1$}, denoted as $A\in$\textbf{$A_1$} if
\begin{align*}
\left\{
\begin{array}{lr}
M_{ij} \geq 0, M_{ii} = -\sum_{j \neq i} ^N M_{i j}, \quad \forall i \neq j \in \Omega;\\
\mathrm{Re}[\lambda (M)] < 0;
\end{array}
\right.
\end{align*}
where $\mathrm{Re}[\cdot] < 0$ means that the real parts of eigenvalues are all negative except an eigenvalue $0$ with the right eigenvector $\textbf{1}_N=(1, 1, \cdots, 1)^T$ and multiplicity $1$.
\end{definition}

\begin{assumption}\label{as1}
Suppose $f(\cdot):R^n \rightarrow R^n$ is continuous and $\mathcal{H} _f$ is a positive constant, then the following condition holds:
\begin{align*}
(x-\tilde{x})^T(f(x)-f(\tilde{x}))\le\mathcal{H} _f(x-\tilde{x})^T(x-\tilde{x}),
\end{align*}
for any vectors $x, \tilde{x} \in R^n$.
\end{assumption}

\begin{lemma}\label{l1}(\cite{RAH})
For any vectors $\xi \in R^N$ and a symmetric matrix $\mathbb{A}  \in R^{N \times N}$, denote $\lambda_{\max}(\mathbb{A} )$ and $\lambda_{\min}(\mathbb{A} )$ as the largest and smallest eigenvalue of $A$, then
\begin{align*}
\lambda_{\min}(\mathbb{A} )\xi^T \xi \leq \xi^T \mathbb{A} \xi \leq \lambda_{\max}(\mathbb{A} )\xi^T \xi.
\end{align*}
\end{lemma}

Next, we present some important lemmas for PT stability.
\begin{lemma}\label{l2}
For a continuous and non-negative real function $V(t)$, suppose the following equation holds:
\begin{align}\label{simplemodel}
\dot{V}(t) = -\frac{\delta V(t)}{\mathcal{C}(t)},
\end{align}
where $\delta>0$ is a constant, $\mathcal{C}(t)$ is a monotonically decreasing function which is nonnegative and contains $T$, and $\mathcal{C}(t)\to 0$ when $t\to T^{-}$.
If
\begin{align}\label{improper}
\int_{0}^T\frac{dt}{\mathcal{C}(t)}=+\infty,
\end{align}
then $V(t)$ will achieve PT stability, i.e., $\lim_{t\to T^{-}}V(t)=0$.
\end{lemma}
\begin{proof}
The model (\ref{simplemodel}) is clearly equivalent to
\begin{align}\label{nec}
\frac{dV(t)}{V(t)} = -\frac{\delta dt}{\mathcal{C}(t)},
\end{align}
where the left hand would be $\ln V(t)-\ln V(0)$, so if PT stability should be realized, then condition (\ref{improper}) should hold, i.e., improper integral should diverge.

If $\mathcal{C}(t)=T-t$, then (\ref{improper}) will become
\begin{align*}
&\int_{0}^T\frac{dt}{T-t}= \ln(T)-\lim_{s\to T}\ln(T-s) = +\infty,
\end{align*}
and for $\mathcal{C}(t)=(T-t)^{\ell}, \ell>1$, (\ref{improper}) would also hold obviously. On the other hand, if $0<\ell<1$, then $\int_{0}^T1/{\mathcal{C}(t)}dt$ would converge, so (\ref{improper}) does not hold.

Then, for model (\ref{simplemodel}), we have
\begin{align*}
\ln V(s) = \ln V(0) - \delta \int_{0}^s\frac{dt}{\mathcal{C}(t)}\Rightarrow V(s) = V(0)e^{-\delta\int_{0}^s\frac{dt}{\mathcal{C}(t)}}\nonumber
\end{align*}
Hence, $V(T)=\lim_{t\to T^{-}}V(t)=0$ can be deduced from (\ref{improper}).  However, $V(T) = 0$ does not mean that its derivative is also zero.
That is to say, $\dot{V}(T)$ may not be zero, or even be infinite at time $T$.
That is why it is necessary to discuss it. In the following, for simplicity, we simplify $t\to T^{-}$ as $t \to T$.

Now, a key discussion about the derivative of $V(t)$ at point $T$ should be given.
If $\mathcal{C}(t)=T-t$, then $V(s)=V(0)(T-t)^\delta/T^{\delta}$, and
\begin{align*}
\varPhi (T) = \lim\limits_{s\to T}\frac{V(s)}{\mathcal{C}(s)} = \lim\limits_{s\to T}V(0)(T-s)^{\delta-1}/T^{\delta},
\end{align*}
therefore, if $\delta<1$, then $\varPhi(T)$ would be infinity; if $\delta=1$, then $\varPhi(T)$ would be a non-zero constant; if $\delta>1$, then $\varPhi(T)=0$.

Otherwise, if $\mathcal{C}(t) = (T-t)^{\ell}, \ell>1$
\begin{align*}
\varPhi(T) = &\lim\limits_{s\to T}\frac{V(s)}{\mathcal{C}(s)}
= V(0)e^{\frac{\delta}{l-1}(T)^{1-\ell}} \lim\limits_{s\to T}\frac{e^{\frac{\delta}{1-l}(T-s)^{1-\ell}}}{(T-s)^{\ell}}
\end{align*}
Based on the Taylor expansion theory, we can also deduce that $\varPhi(T)=0$, so $\dot{V}(T) = 0$. Hence, PT stability is obtained.

For other forms of $\mathcal{C}(t)$, we can expand it by the Taylor series as the powers of $(T-t)$. The proof is completed.
\end{proof}

An example for the dynamics of (\ref{simplemodel}) with  $\mathcal{C}(t)=(1-t)^{\ell}$ are shown in Fig. \ref{se} with $\ell=1,2,3$, $\delta =0.5, 1, 2$, and $V(0) = 15$. For fixed $\ell$, the larger $\delta$ is, the higher the convergence speed is (but the settling times are the same), see the three red lines and blue lines. On the other hand, for fixed $\delta$, the larger $\ell$ is, the higher the convergence speed is, see the three solid lines.

\begin{figure}[h]
    \centering
    \includegraphics[height=0.4\textheight,width=0.7\textwidth] {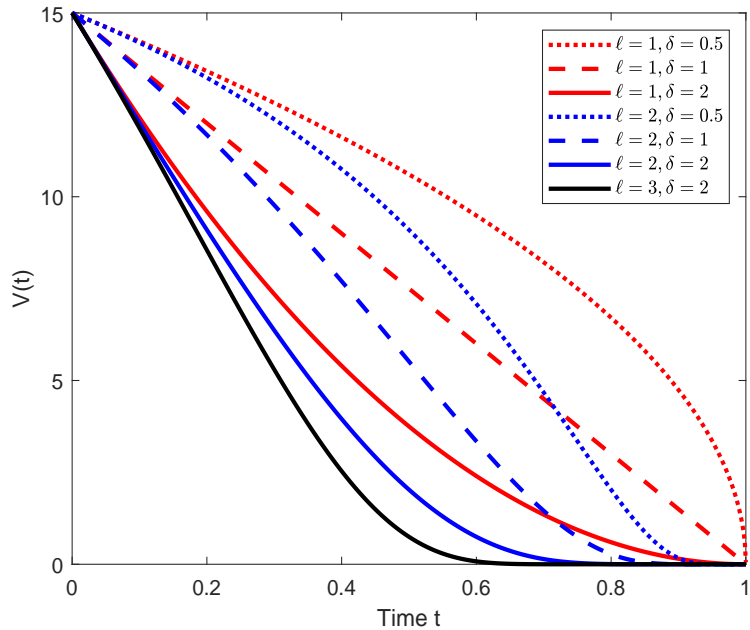}
    \caption{Different dynamics of $V(t)$ under different $\ell$ and $\delta$. One should notice: when $\ell=1$, if $\delta\le 1$, see the first and the second red lines, the derivative of $V(t)$ at time $T$ would be infinity or a non-zero constant, whereas if $\ell>1$, then no matter $\delta$ is, the derivative of $V(t)$ at time $T$ would be zero.}\label{se}
\end{figure}

From the aforementioned discussions, we can deduce that the number of the time-varying function $\mathcal{C}(t)$ is infinite. In lots of previous papers which investigate PT stability, they have constructed different forms of $\mathcal{C}(t)$. Here we summarize them as two categories: one contains exponential functions, whereas the other contains power functions. For the former, in \cite{ZB1}, $1/\mathcal{C}(t)=\frac{e^{aT} - 1}{e^{aT} - e^{at}}(a > 0)$, so
\begin{align*}
\int_{0}^T\frac{dt}{\mathcal{C}(t)}&=\int_{0}^T\frac{e^{aT} - 1}{e^{aT} - e^{at}}dt=\frac{e^{aT}-1}{e^{aT}}\int_{0}^T\frac{1}{1- e^{a(t-T)}}dt\\
&=\frac{e^{aT}-1}{ae^{aT}}\bigg[\lim_{s\to T}\ln\frac{e^{a(s-T)}}{1-e^{a(s-T)}}-\ln\frac{e^{-aT}}{1-e^{-aT}}\bigg]=+\infty,
\end{align*}
and in \cite{WZW1}, $1/\mathcal{C}(t)=\frac{ae^{a(T - t)}}{e^{a(T - t)} - 1}(a > 0)$, which can be integrated as above except the coefficient. For the latter, there are many forms of $(T - t)^\ell$, for example, $\ell = 1$ in \cite{lxy1} and \cite{lxy2}; $\ell \geq 1$ in \cite{SA1}; $\ell = 2$ in \cite{RYH1}; $\ell \geq 2$ in \cite{LWQ1}; $\ell = 4$ in \cite{SYD4} and \cite{SYD5}. All of them satisfy $\int_{0}^T 1/\mathcal{C}(t)dt = +\infty$, so that PT stability is achieved according to Lemma \ref{l2}.

\begin{remark}
In fact, a more general model can be set up:
\begin{align}\label{secondmodel}
\dot{V}(t) = -\frac{\delta V^p(t)}{\mathcal{C}(t)}, ~~V(t)\ge 0, ~~\mathcal{C}(t)> 0,
\end{align}
when $p=1$, it becomes (\ref{simplemodel}). Otherwise,
\begin{align*}
V(s)^{-p+1}-V(0)^{-p+1}=(p-1)\int_0^s\frac{\delta}{\mathcal{C}(t)}dt
\end{align*}
Therefore, if $p>1$, then the improper integral (\ref{nec}) should diverge, and PT stability can be obtained with the similar analysis as the above lemma. On the other hand, if $0<p<1$, then integral (\ref{nec}) should converge, for example, $\mathcal{C}(t)$ is a constant, then it becomes the finite-time stability like \cite{LLC2016}. Therefore, we still can consider the finite-time or fixed-time stability over prescribed-time interval, and in this case improper integral may become normal integral.
\end{remark}

Next, we consider a more complicated but more useful model, and consider its PT stability.
\begin{lemma}\label{ll5}
Suppose there is a continuous and non-negative real function $V(t)$ with
\begin{align}\label{thirdmodel}
\dot{V}(t)=\delta_1 V^p(t) - \frac{\delta_2V(t)}{\mathcal{C}(t)},
\end{align}
where $\delta_1\ge 0, \delta_2>0$, $p > 0$ and $\mathcal{C}(t)$ is defined in Lemma \ref{l2}. If (\ref{improper}) holds, then PT stability can be achieved with large $\delta_2$.
\end{lemma}

\begin{proof}
If $p=1$, model (\ref{thirdmodel}) can be written as
\begin{align*}
\dot{V}(t)=\bigg(\delta_1-\frac{\delta_2}{\mathcal{C}(t)}\bigg)V(t)=-\frac{\delta_2}{\mathcal{C}^{\prime}(t)}V(t)
\end{align*}
where ${\mathcal{C}^{\prime}(t)}={\delta_2\mathcal{C}(t)}/(\delta_2-\delta_1\mathcal{C}(t))$, so this model has been investigated in the above lemma.

Hence, we consider $p\neq 1$ in the following. Let $\bar{V}(t)=V^{1-p}(t)$ at first, then we have
\begin{align*}
\dot{\bar{V}}(t)=\delta_1(1-p)-\frac{\delta_2(1-p)}{\mathcal{C}(t)}\bar{V}(t),
\end{align*}
This first-order differential equation can be solved, i.e.,
\begin{align}\label{iny}
V(s)=e^{-\int_0^s\frac{\delta_2}{\mathcal{C}(t)}dt}[\int_0^s (1 - p)\delta_1 e^{\int_0^t\frac{\delta_2(1 - p)}{\mathcal{C}(x)}dx}dt
+ V^{1 - p}(0)]^{\frac{1}{1 - p}}
\end{align}

Next, we will discuss PT stability in two steps. The first step is to investigate the value of ${V}(t)$ at time $T$. If $0 < p < 1$,
\begin{align*}
\lim_{s\to T} e^{-\int_0^s\frac{\delta_2(1 - p)}{\mathcal{C}(t)}dt} V^{1 - p}(0) = 0
\end{align*}
Moreover, by using L'Hospital rule, (\ref{iny}) will become
\begin{align}
{V}^{1 - p}(T) =\lim_{s\to T}\frac{\int_0^s(1 - p)\delta_1 e^{\int_0^t\frac{\delta_2(1 - p)}{\mathcal{C}(r)}dr}dt}{e^{\int_0^s\frac{\delta_2(1 - p)}{\mathcal{C}(t)}dt}}=\lim_{s\to T}\frac{(1 - p)\delta_1 e^{\int_0^s\frac{\delta_2(1 - p)}{\mathcal{C}(r)}dr}}{\frac{\delta_2(1 - p)}{\mathcal{C}(s)}e^{\int_0^s\frac{\delta_2(1 - p)}{\mathcal{C}(t)}dt}}
= \lim_{s\to T}\frac{\delta_1\mathcal{C}(s)}{\delta_2} = 0.\nonumber
\end{align}

Otherwise, if $p>1$, we know
\begin{align*}
\lim_{s\to T}\int_0^s (1 - p)\delta_1 e^{\int_0^t\frac{\delta_2(1 - p)}{\mathcal{C}(r)}dr}dt = 0
\end{align*}
Then, (\ref{iny}) will become
\begin{align*}
{V}^{1 - p}(T) = \lim_{s\to T}e^{\int_0^s\frac{\delta_2(p - 1)}{\mathcal{C}(t)}dt}V^{1 - p}(0) = +\infty
\end{align*}
Therefore, no matter the value of $p$ is, ${V}(T) = 0$.

The second step is to investigate the value of $\dot{{V}}(t)$ at $T$.
\begin{align*}
\dot{V}(T) = \lim_{s\to T}(\delta_1 V^p(s) - \frac{\delta_2V(s)}{\mathcal{C}(s)}) = -\delta_2\lim_{s\to T}\frac{V(s)}{\mathcal{C}(s)}
\end{align*}
Similar to Lemma \ref{l2}, it is clear that whether $\mathcal{C}(t) = T - t$, $(T - t)^\ell(\ell > 1)$ or other effective forms, $\dot{{V}}(T) = 0$ holds. Therefore, PT stability would be achieved.
\end{proof}

\begin{remark}
According to the Comparison Theorem, PT stability can also be achieved for the case that `$=$' is replaced by `$\le$' in equations (\ref{simplemodel}), (\ref{secondmodel}) and (\ref{thirdmodel}). Therefore, in the previous discussion to deal with the case $p=1$, one simpler method is directly to choose $\delta_2=\mathcal{C}(0)\delta_1+2$, then (\ref{secondmodel}) would become $\dot{V}(t)\le- \frac{2V(t)}{\mathcal{C}(t)}$, and according to Lemma \ref{l2}, we can also obtain the PT stability.
\end{remark}

\begin{remark}
The application of the above model includes the case that a function cannot satisfy the Lipschitz condition but with higher nonlinearities. Of course, we just consider a fraction of all combinations of parameters $p$ and $q$ for the general model: $\dot{V}(t)\leq \delta_1 V^p(t) - \frac{\delta_2 V^q(t)}{\mathcal{C}(t)}, p > 0, q > 0$. Interested readers are encouraged to study the other cases.
\end{remark}

\section{PT synchronization and control for MWDCNs}\label{MWD}
In this section, we will apply our obtained lemmas on the synchronization and control problems for MWDCNs.

\subsection{PT synchronization without control}
In this subsection, to realize PT synchronization only by mutual coupling, we design the network model as
\begin{align}\label{model0}
\dot{x}_i(t) = f(x_i(t)) + \frac{\eta}{\mathcal{C}(t)} \sum_{w = 1} ^ W \sum_{j = 1} ^ N M_{ij}^w \Gamma^w x_j(t)
\end{align}
where $x_i(t)=(x_i^1(t), x_i^2 (t), \cdots, x_i^n(t))^T \in R^n$
is the state of the $i$-th node; function $f(\cdot ):R^n \rightarrow R^n$ is continuous with Assumption \ref{as1}; $W>1$ means the multiple network topologies; and $\eta$ means the coupling strength. $M^w=(M^w_{ij})_{N \times N}$ stands for asymmetric OCMs with zero-row-sum, i.e., $\sum_{j=1}^NM^w_{ij}=0, \forall i$. Symmetric ICMs $\Gamma^w = (\gamma^w_{ij})\in R^{n\times n}$.

Next, we define the PT synchronization for (\ref{model0}).
\begin{definition}\label{pts}
Complex network (\ref{model0}) is said to realize PT synchronization globally if there exists a controller such that
\begin{align}\label{df2}
\lim_{t\rightarrow T}\|x_i(t)-x_j(t)\|_2 = 0
\end{align}
for all $i,j \in \Omega $, and $T$ is independent of any parameter.
\end{definition}

For convenience of calculations, we define the regulator $\mathcal{C}(t)$ for PT synchronization as
\begin{align}\label{important}
\mathcal{C} (t) = (T - t)^\ell, ~~~~~~t \in [0, T),~~ \ell \ge1
\end{align}

The synchronization problem is investigated by considering each dimension of the state of all nodes separately, which is called rearranging variables' order technique (ROT) in \cite{lxw1}, and it can transform the MWDCNs into networks with a single weight for each dimension.

Therefore, by considering each dimension separately, we consider the dynamics of $x^d_i(t)$, which can be depicted as
\begin{align*}
\dot{x}_i^d(t)=f(x_i(t))^d+\frac{\eta}{\mathcal{C}(t)}\sum_{w=1}^W\sum_{j=1}^N M_{ij} ^w \bigg({\gamma}^w _{dd} x_j^d(t)+\sum_{e\neq d}{\gamma}^w_{de} x_j^e(t)\bigg),
\end{align*}
where $f(x_i(t))^d$ is the $d$-th dimension of $f(x_i(t)), d=1,2,\cdots,n$. For any dimension $d$, we denote
\begin{align}
&\mathcal{X}^{[d]}(t)=(x_1^d(t), x_2^d(t), \cdots, x_N^d(t))^T,\\
&\mathcal{F}^{[d]}(t)=(f(x_1(t))^d, f(x_2(t))^d,\cdots,f(x_N(t))^d)^T,
\end{align}
then
\begin{align}\label{transform}
\dot{\mathcal{X}}^{[d]}(t)=\mathcal{F}^{[d]}(t)+
\frac{\eta}{\mathcal{C}(t)}(M^{[dd]}\mathcal{X}^{[d]}(t)+\sum_{e\ne d}M^{[de]}\mathcal{X}^{[e]}(t)),
\end{align}
where sum (union) matrices mentioned above are denoted as
\begin{align}\label{s1}
M^{[de]}=\sum_{w=1}^W {\gamma}_{de}^w M^w, ~~~~d, e=1,2,\cdots,n.
\end{align}

Since all $M^w$ are zero-row-sum matrices, $M^{[dd]}$ and $M^{[de]}$ defined in (\ref{s1}) are also zero-row-sum matrices. Furthermore,
\begin{assumption}\label{sc}
We assume $M^{[dd]}\in A_1, d=1,2,\cdots,n$.
\end{assumption}

Obviously, we just require the above assumption, so there is no requirement on single OCM $M^w$ and ICM $\Gamma^w$, elements $M_{ij}^w$ and $\gamma_{ij}^w$ can be positive, negative, or zero. Hence, this condition is more general and widely applicable than \cite{lxw1}.

Based on the results in \cite{WT1}, under Assumption \ref{sc}, the normalized left eigenvector (NLEVec) corresponding to eigenvalue $0$ of $M^{[dd]}$ exists, i.e., $(\psi^{[d]})^TM^{[dd]}=0$, where
\begin{align}\label{nlevec}
\psi^{[d]}=(\psi^{[d]}_1,\cdots,\psi^{[d]}_N)^T\in R^N,
\end{align}
with $\sum_{i=1}^N \psi_i^{[d]}= 1$ and $\psi_i^{[d]}>0, i=1,2,\cdots, N$.
Define $\Psi^{[d]}=\mathrm{diag}(\psi^{[d]})$, then symmetric matrices $\mathcal{I}_{\psi^{[d]}}=\Psi^{[d]}-\psi^{[d]}(\psi^{[d]})^T\in A_1$. Denote
\begin{align}
\mathcal{I} _\Psi=&\mathrm{diag}(\mathcal{I}_{\psi^{[1]}}, \mathcal{I}_{\psi^{[2]}}, \cdots, \mathcal{I}_{\psi^{[n]}}),\label{ipsi}\\
M=&\begin{pmatrix}
M^{[11]}&\cdots&M^{[1n]}\\
\vdots&\ddots&\vdots\\
M^{[n1]}&\cdots&M^{[nn]}
\end{pmatrix}.\label{expand}
\end{align}

\begin{theorem}\label{PT-syn}
For the MWDCN (\ref{model0}) with $\mathcal{C}(t)$ defined in (\ref{important}), Assumption \ref{as1} and Assumption \ref{sc} hold, if the following matrix
\begin{align}\label{expand-syn}
\overline{M}=(\mathcal{I} _\Psi M+M^T\mathcal{I} _\Psi)/2
\end{align}
is negative definite in the transverse space $\mathcal{TS}=\{\mathcal{R}\in R^N|\mathcal{R}^T\textbf{1}=0\}$, where $\textbf{1}=(1,\cdots,1)^T$, then PT synchronization can be realized with $\eta>(\mathcal{H}_f\mathcal{C}(0)+1)/|\lambda_2(\overline{M})|$, where $\lambda_2(\cdot)$ signifies the second largest eigenvalue in the whole space which is also called {Fiedler eigenvalue} \cite{RR}.
\end{theorem}

\begin{proof}
Using NLEVec $\psi^{[d]}$ in (\ref{nlevec}), we can define dummy synchronization targets as $x^d_{\star}(t)=\sum_{i=1}^N\psi_i^{[d]}x_i^d(t)$. Thus, the PT synchronization in Definition \ref{pts} is equivalent to prove that
$\lim_{t\to T}|x_i^d(t)-x^d_{\star}(t)|=0, d=1,2,\cdots,n$.

Define the Lyapunov function
\begin{align}\label{lyps1}
\mathcal{W}(t) = \frac{1}{2} \sum_{i = 1}^N \sum_{d = 1}^n \psi_i^{[d]}(x_i^d(t) - x^d_{\star}(t))^2= \frac{1}{2} \sum_{d = 1}^n \mathcal{X}^{[d]}(t)^T \mathcal{I}_{\psi^{[d]}} \mathcal{X}^{[d]}(t)
\end{align}
Thus, based on Lemma \ref{l1}, $\dot{\mathcal{W}}(t)$ along (\ref{transform}) satisfies:
\begin{align}
\dot{\mathcal{W}}(t) =& \sum_{d = 1}^n \mathcal{X}^{[d]}(t)^T \mathcal{I}_{\psi^{[d]}} \dot{\mathcal{X}}^{[d]}(t)\le2\mathcal{H}_f\mathcal{W}(t)+ \frac{\eta}{\mathcal{C}(t)}\mathcal{X}(t)^T\overline{M}\mathcal{X}(t)\nonumber\\
\le&2\mathcal{H}_f\mathcal{W}(t)+ 2\lambda_{2}(\overline{M})\frac{\eta}{\mathcal{C}(t)}\mathcal{W}(t),\label{cou}
\end{align}
where $\mathcal{X}(t)=(\mathcal{X}^{[1]}(t)^T,\cdots,\mathcal{X}^{[n]}(t)^T)^T$.

The form (\ref{cou}) is the same as that in Lemma \ref{ll5} if we set $\delta_1= 2\mathcal{H}_f$ and $\delta_2=-2\lambda_{2}(\overline{M}){\eta}$. According to Lemma \ref{ll5}, if $\eta$ is large enough, PT stability for $\mathcal{W}(t)$ holds, that is to say, $\mathcal{W}(t)$ would converge to zero when $t\to T$ and the derivative of $\mathcal{W}(t)$ is also zero. Therefore, from the concrete form of $\mathcal{W}(t)$ defined in (\ref{lyps1}), PT synchronization is realized.
\end{proof}

Recalling the prerequisite that matrix $\overline{M}$ defined in (\ref{expand-syn}) should be negative definite in the transverse space, we know that the condition cannot be ensured for any matrix, therefore, in the following, we will consider a special case: all ICMs are diagonal, i.e., $\Gamma^w=\mathrm{diag}(\gamma^w_{11},\cdots,\gamma^w_{nn})$. In this case, sum matrices defined in (\ref{s1}) satisfy: $M^{[de]}=0$ for any $d\ne e$, and the matrix $M$ defined in (\ref{expand}) would be: $M=\mathrm{diag}(M^{[11]},\cdots,M^{[nn]})$. According to \cite{WT1}, matrices
\begin{align}\label{zuhe}
_{\psi^{[d]}} M^{[dd]} = \frac{\Psi^{[d]} M^{[dd]} + (M^{[dd]})^T \Psi^{[d]}}{2}
\end{align}
are symmetric and negative definite in $\mathcal{TS}$, which means that matrix $\overline{M}$ defined in (\ref{expand-syn}) is surely negative definite in $\mathcal{TS}$.

\begin{theorem}\label{diag-syn}
For the MWDCN (\ref{model0}) with diagonal ICMs, suppose Assumption \ref{as1} and Assumption \ref{sc} hold, then PT synchronization can be realized with $\eta>(\mathcal{H}_f\mathcal{C}(0)+1)/|\lambda_2(\overline{M})|$.
\end{theorem}

\subsection{PT synchronization with pinning control}
In this subsection, we will consider the pinning control problem, whose network model can be described as
\begin{align}\label{model1}
\dot{x}_i(t)=f(x_i(t))+\frac{\eta}{\mathcal{C}(t)}\sum_{w=1}^W \sum_{j=1}^N M_{ij}^w \Gamma^w x_j(t)-\kappa_i\frac{\eta}{\mathcal{C}(t)}\Gamma(x_i(t)-x_0(t));
\end{align}
where the meaning of parameters is the same with those in model (\ref{model0}), and $\kappa _i=1$ if $i=1$ (the first node is pinned), otherwise $\kappa_i=0$; symmetric matrix $\Gamma=(\gamma_{ij})$ can be independent of ICMs $\Gamma^w$; $x_0(t)$ is the control target satisfying
\begin{align}\label{st}
\dot{x}_0(t)=f(x_0(t)).
\end{align}

Next, we define the PT synchronization for (\ref{model1}).
\begin{definition}\label{ptc}
Complex network (\ref{model1}) is said to realize PT synchronization globally if there exists a controller such that
\begin{align}\label{df3}
\lim_{t\rightarrow T}\|x_i(t) - x_0(t)\|_2 = 0
\end{align}
for all $i,j \in \Omega $, and $T$ is independent of any parameter.
\end{definition}

For any dimension $d=1, 2, \cdots, n$, we denote
\begin{align}
&\mathcal{Y}^{[d]}(t)=(x_1^d(t)-x_0^d(t), \cdots, x_N^d(t)-x_0^d(t))^T,\\
&\hat{\mathcal{F}}^{[d]}(t)=(f(x_1(t))^d-f(x_0(t))^d,\cdots,f(x_N(t))^d-f(x_0(t))^d)^T,
\end{align}
then
\begin{align}\label{transform2}
\dot{\mathcal{Y}}^{[d]}(t)=\hat{\mathcal{F}}^{[d]}(t)+
\frac{\eta}{\mathcal{C}(t)}(\hat{M}^{[dd]}\mathcal{Y}^{[d]}(t)+\sum_{e\ne d}\hat{M}^{[de]}\mathcal{Y}^{[e]}(t)),
\end{align}
where sum (union) matrices mentioned above are denoted as
\begin{align}\label{s2}
\hat{M}^{[de]}=M^{[de]}-\mathrm{diag}(\gamma_{de},0,\cdots,0),
\end{align}
where $M^{[de]}$ is defined in (\ref{s1}).

According to \cite{CTP1}, under Assumption \ref{sc}, the NLEVec $\psi^{[d]}$ of $M^{[dd]}$ exists, which can make the new matrices
\begin{align}\label{zuhe1}
_{\psi^{[d]}} \hat{M}^{[dd]} = \frac{\Psi^{[d]} \hat{M}^{[dd]} + (\hat{M}^{[dd]})^T \Psi^{[d]}}{2}
\end{align}
be symmetric and negative definite, where $\Psi^{[d]}=\mathrm{diag}(\psi^{[d]})$.

\begin{theorem}
For the MWDCNs (\ref{model1}) with $\mathcal{C}(t)$ defined in (\ref{important}), Assumption \ref{as1} and \ref{sc} hold, if the following matrix
\begin{align}\label{expand-cont}
\widetilde{M}=(\Psi \hat{M}+\hat{M}^T\Psi)/2
\end{align}
is negative definite, where $\Psi=\mathrm{diag}(\Psi^{[1]},\cdots,\Psi^{[n]})$, and
\begin{align}
\hat{M}=&\begin{pmatrix}
\hat{M}^{[11]}&\cdots&\hat{M}^{[1n]}\\
\vdots&\ddots&\vdots\\
\hat{M}^{[n1]}&\cdots&\hat{M}^{[nn]}
\end{pmatrix},\label{expand2}
\end{align}
then PT synchronization can be realized with $\eta>(\mathcal{H}_f\mathcal{C}(0)+1)/|\lambda_{\max}(\widetilde{M})|$, where $\lambda_{\max}(\cdot)$ signifies the largest eigenvalue.
\end{theorem}

\begin{proof}
PT synchronization in Definition \ref{ptc} is equivalent to prove that
$\lim_{t\to T}|x_i^d(t)-x^d_{0}(t)|=0, d=1,2,\cdots,n$.
Therefore, we choose the Lyapunov function as
\begin{align}\label{lyp1}
\mathcal{W}(t)=\frac{1}{2}\sum_{i=1}^N\sum_{d=1}^n \psi_i^{[d]}(x_i^d(t)-x_0^d(t))^2=\frac{1}{2}\sum_{d=1}^n \mathcal{Y}^{[d]}(t)^T\Psi^{[d]}\mathcal{Y}^{[d]}(t)
\end{align}
Thus, based on Lemma \ref{l1}, $\dot{\mathcal{W}}(t)$ along (\ref{transform2}) satisfies:
\begin{align}
\dot{\mathcal{W}}(t) =& \sum_{d=1}^n \mathcal{Y}^{[d]}(t)^T\Psi^{[d]}\dot{\mathcal{Y}}^{[d]}(t)\le2\mathcal{H}_f\mathcal{W}(t)+ \frac{\eta}{\mathcal{C}(t)}\mathcal{Y}(t)^T\widetilde{M}\mathcal{Y}(t)\nonumber\\
\le&2\mathcal{H}_f\mathcal{W}(t)+ 2\lambda_{\max}(\widetilde{M})\frac{\eta}{\mathcal{C}(t)}\mathcal{W}(t),\label{cou2}
\end{align}
where $\mathcal{Y}(t)=(\mathcal{Y}^{[1]}(t)^T,\cdots,\mathcal{Y}^{[n]}(t)^T)^T$.

With the same arguments as those in Theorem \ref{PT-syn}, we can also get the PT synchronization for the MWDCN (\ref{model1}).
\end{proof}

Similarly, when all ICMs are diagonal, i.e., $\Gamma^w=\mathrm{diag}(\gamma^w_{11},\cdots,\gamma^w_{nn})$ and $\Gamma=\mathrm{diag}(\gamma_{11},\cdots,\gamma_{nn})$, $\hat{M}^{[de]}=0, d\ne e$, therefore, $\hat{M}=\mathrm{diag}(\hat{M}^{[11]},\cdots,\hat{M}^{[nn]})$, according to (\ref{zuhe1}), the matrix $\widetilde{M}$ would be negative definite surely. Therefore, we have the following result.

\begin{theorem}\label{diag-ctrl}
For the MWDCN (\ref{model1}) with diagonal ICMs, suppose Assumption \ref{as1} and \ref{sc} hold, then PT synchronization can be realized with $\eta>(\mathcal{H}_f\mathcal{C}(0)+1)/|\lambda_{\max}(\widetilde{M})|$.
\end{theorem}

\begin{remark}
The synchronization for complex networks in a finite-time or fixed-time is proved in \cite{LLC2016,lxw2}. However, the settling time depends heavily on the initial system states and/or coupling coefficients and control parameters. It cannot be preset by users according to the needs of tasks. Besides, \cite{lxy1,lxy2} only considered the network with a single weight. Thus, PT synchronization and control for MWDCNs is investigated in our work, which is more realistic and challengeable.
\end{remark}

\begin{remark}
Different from finite-time or fixed-time synchronization, which can be achieved by using nonlinear coupling/control functions, the above designed PT synchronization protocol only uses linear coupling/control functions, and the multiple coupling matrices can be asymmetric, so it is simpler in real applications, whereas the main difference from classical synchronization \cite{WT1} is the design of coupling strength.
\end{remark}

\begin{remark}
Notice that we only discuss the convergence on the time interval $[0, T)$, if one wants to continue to discuss the dynamics after time $T$, for example, keeping the error to be zero, then one can follow and adopt previous studies to maintain synchronization, here we omit these discussions.
\end{remark}

\section{Numerical Simulation}\label{simulation}
Consider a MWDCN (\ref{model0}) with three nodes and four weights, and $f(\cdot)$ is described by \cite{ZF1}:
\begin{align*}
f({x}_i(t))=-x_i(t) +
    \begin{pmatrix}
        -1.25 &-3.2& -3.2\\
        -3.2 &1.1& -4.4\\
        -3.2 &4.4& 1
    \end{pmatrix}
    \begin{pmatrix}
        h(x_i^1(t))\\
        h(x_i^2(t))\\
        h(x_i^3(t))
    \end{pmatrix},
\end{align*}
where $x_i(t)=(x_i^1(t),x_i^2(t),x_i^3(t))^T$
and $h(x_i^d(t))=(|x_i^d(t)+1|-|x_i^d(t)-1|)/2$,
$d=1,2,3$. Hence, $\mathcal{H}_f=5.4704$.

OCMs are chosen as:
\begin{align*}
M^1 =\begin{pmatrix}
        -3 &3& 0\\
        0 &0& 0\\
        3 &0& -3
    \end{pmatrix},
M^2 =\begin{pmatrix}
        0 &0& 0\\
        0 &-6& 6\\
        3 &0& -3
    \end{pmatrix},
M^3 =\begin{pmatrix}
        2 &-2& 0\\
        0 &0& 0\\
        4 &0& -4
    \end{pmatrix},
M^4=\begin{pmatrix}
        -2 &0& 2\\
        0 &-5& 5\\
        4 &0& -4
    \end{pmatrix}.
\end{align*}
Obviously, $M^1$, $M^2$ and $M^3$ represent the network can be not strongly connected, and elements can have competitive relationships. ICMs are chosen as: $\Gamma^1= \mathrm{diag}(7, 5, 6), \Gamma^2=\mathrm{diag}(7, 5, 6), \Gamma^3=\mathrm{diag}(6, -1, 1)$, and $\Gamma^4=\mathrm{diag}(6, 5, 7)$.

Using the formula (\ref{s1}), the sum matrices are
\begin{align*}
M^{[11]} =
    \begin{pmatrix}
        -21 &9& 12\\
        0 &-72& 72\\
        90 &0& -90
    \end{pmatrix},
    M^{[22]} =
    \begin{pmatrix}
        -27 &17& 10\\
        0 &-55& 55\\
        46 &0& -46
    \end{pmatrix},
    M^{[33]} =
    \begin{pmatrix}
        -30 &16& 14\\
        0 &-71& 71\\
        68 &0& -68
    \end{pmatrix}.
\end{align*}
Therefore, Assumption \ref{sc} holds, and the corresponding NLEVec are: $\psi^{[1]}=(0.7362, 0.0920, 0.1718)^T$, $\psi^{[2]}=(0.5274, 0.163, 0.3096)^T$, $\psi^{[3]}=(0.6, 0.1352, 0.2647)^T$. Therefore, $\lambda_2(\overline{M})=-9.9387$. Regulator $\mathcal{C}(t)=(3-t)^2$. According to Theorem \ref{diag-syn}, if $\eta>(\mathcal{H}_f\mathcal{C}(0)+1)/|\lambda_2(\overline{M})|=2.7222$, then PT synchronization can be realized.

The index $E_1(t)=\|x_2(t)-x_1(t)\|_2+\|x_3(t)-x_1(t)\|_2$ is used to denote the synchronization error between nodes. Fig. \ref{simu-pts} shows the dynamics of $E_1(t)$ with $\eta=0.35$, where the initial values are: $x_1(0) = (10, 15, 20)^T$, $x_2(0) = (25, 30, 35)^T$, $x_3(0) = (40, 45, 50)^T$.

\begin{figure}
\centering
\includegraphics[height=0.4\textheight,width=0.7\textwidth] {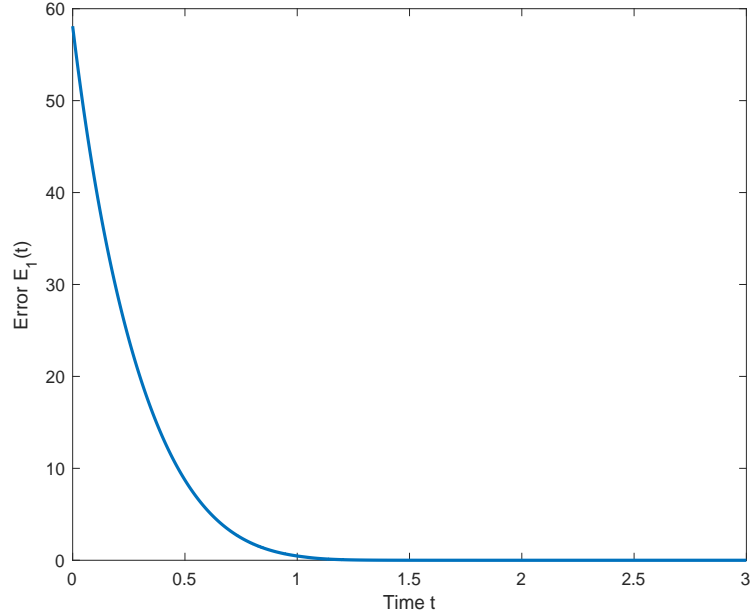}
\caption{Dynamics of $E_1(t)$ under $\eta=0.35$, which shows that PT synchronization can be realized.}\label{simu-pts}
\end{figure}

On the other hand, we consider PT synchronization with pinning control for a MWDCN (\ref{model1}), where the parameters are the same as defined above, the pinning control is added on the first node, and $\Gamma=\mathrm{diag}(11, 13, 15)$, therefore, Using the formula (\ref{s2}), the sum matrices are
\begin{align*}
\hat{M}^{[11]} =
    \begin{pmatrix}
        -32 &9& 12\\
        0 &-72& 72\\
        90 &0& -90
    \end{pmatrix},
\hat{M}^{[22]} =
    \begin{pmatrix}
        -40 &17& 10\\
        0 &-55& 55\\
        46 &0& -46
    \end{pmatrix},
\hat{M}^{[33]} =
    \begin{pmatrix}
        -45 &16& 14\\
        0 &-71& 71\\
        68 &0& -68
    \end{pmatrix}.
\end{align*}
Hence, $\lambda_{\max}(\widetilde{M})=-1.8241$. According to Theorem \ref{diag-ctrl}, if $\eta>(\mathcal{H}_f\mathcal{C}(0)+1)/|\lambda_{\max}(\widetilde{M})|=27.5386$, then PT synchronization can be realized. Denote the synchronization error as $E_2(t)=\sum_{i=1}^3\|x_i(t)-x_0(t)\|_2$, where $x_0(t)$ is the synchronization target satisfying (\ref{st}) with initial value $x_0(0)=(4, 8, 12)^T$. Fig. \ref{simu-ptc} shows the dynamics of $E_2(t)$ with $\eta=0.35$.

\begin{figure}
\centering
\includegraphics[height=0.4\textheight,width=0.7\textwidth] {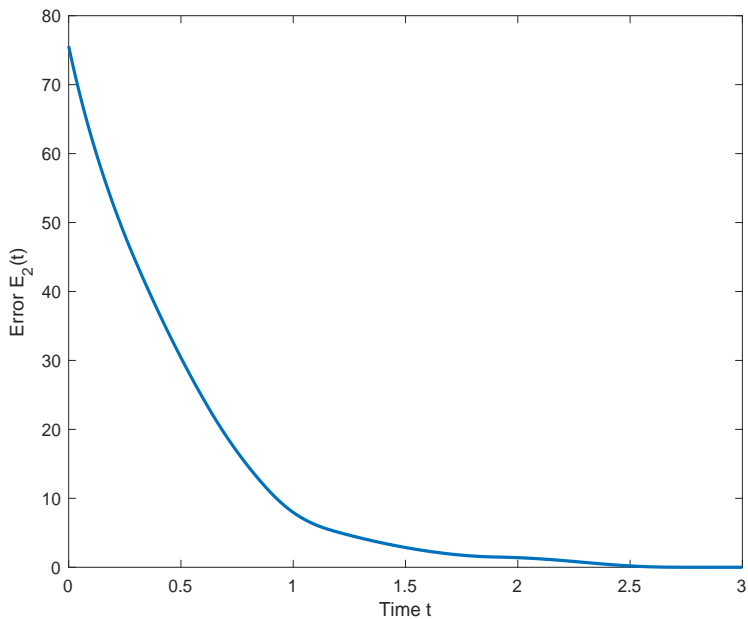}
\caption{Dynamics of $E_2(t)$ under $\eta=0.35$, which shows that PT synchronization can be realized.}\label{simu-ptc}
\end{figure}

\section{Conclusion}\label{conclusion}
In this note, we first study the PT stability, whose settling time is independent of initial values and system parameters. Divergency of improper integral for the time-varying regulator function is used to illustrate the essence of PT stability, and two useful lemmas are proposed. The regulator functions in many previous papers are included in our theory. We also prove rigorously that its derivative is zero, which is important for real applications, because it reflects the magnitude of control. Based on these results, we investigate the PT synchronization for multiweighted and directed complex networks with or without pinning control. Both cooperative and competitive relationships between nodes are allowed. Symmetric ICMs can be diagonal or non-diagonal. We use ROT to consider the PT synchronization for each dimension separately, thus the multiweighted complex networks can be transformed into a single-weighted complex network. At last, numerical simulations are given to verify the theoretical results.


\end{document}